\documentclass[12pt]{amsart} % Default font size is 12pt, it can be changed here

\usepackage[letterpaper, margin=1in]{geometry} % Required to change the page size to A4
\usepackage{times}

\usepackage{amsaddr}

\usepackage[utf8]{inputenc} % allow utf-8 input
\usepackage[T1]{fontenc}    % use 8-bit T1 fonts
\usepackage{hyperref}       % hyperlinks
\usepackage{url}            % simple URL typesetting
\usepackage{booktabs}       % professional-quality tables
\usepackage{nicefrac}       % compact symbols for 1/2, etc.
\usepackage{microtype}      % microtypography
\usepackage{algorithmicx}
\usepackage[ruled,vlined]{algorithm2e}
\usepackage{afterpage}
\usepackage{graphicx}
\usepackage{epstopdf}
\usepackage{caption}
\usepackage{color}
\usepackage[multiple]{footmisc}
\usepackage{subfigure}
\usepackage[usenames,dvipsnames]{xcolor}
\usepackage{listings}
\usepackage{wrapfig}

\usepackage{amsthm}
\newtheorem{theorem}{Theorem}
\newtheorem{proposition}{Proposition}
\newtheorem{question}{Question}
\newtheorem{lemma}{Lemma}

\newtheorem{corollary}{Corollary}
\newtheorem{observation}{Observation}

\newcommand{\xhdr}[1]{\vspace{1.5mm}\noindent{{\bf #1}}}

\newcommand{\sizeC}{\lvert C \rvert}

\definecolor{mylinkcolor}{RGB}{0,0,140}
\hypersetup{colorlinks,allcolors=mylinkcolor,citecolor=mylinkcolor}
\usepackage[capitalize]{cleveref}
\crefname{observation}{Observation}{Observations}

\usepackage{amssymb}

\let\emptyset\varnothing

\usepackage{paralist}
\renewenvironment{itemize}[1]{\begin{compactitem}#1}{\end{compactitem}}
\renewenvironment{enumerate}[1]{\begin{compactenum}#1}{\end{compactenum}}

\begin{document}

\title{Found Graph Data and Planted Vertex Covers}

\author{Austin R.~Benson}
\address{Cornell University}
\email{arb@cs.cornell.edu}

\author{Jon Kleinberg}
\address{Cornell University}
\email{kleinber@cs.cornell.edu}

%!TEX root = found-graph-data-paper.tex

\begin{abstract}
A typical way in which network data is recorded is to measure
all the interactions among a specified set of {\em core nodes}; this
produces a graph containing this core together with a potentially
larger set of {\em fringe nodes} that have links to the core.
Interactions between pairs of nodes in the fringe, however, are
not recorded by this process, and hence not present in the 
resulting graph data.
For example, a phone service provider may only have records 
of calls in which at least one of the participants is a customer;
this can include calls between a customer and a non-customer, but
not between pairs of non-customers.

Knowledge of which nodes belong to the core is an important piece
of metadata that is crucial for interpreting the network dataset.
But in many cases, this metadata is not available, either because
it has been lost due to difficulties in data provenance, or 
because the network consists of ``found data'' obtained in settings
such as counter-surveillance.
This leads to a natural algorithmic problem, namely the recovery of
the core set.
Since the core set forms a vertex cover of the graph, we essentially
have a {\em planted vertex cover} problem, but with an arbitrary
underlying graph.
We develop a theoretical framework for analyzing this planted
vertex cover problem, based on results in the theory of
fixed-parameter tractability,
together with algorithms for recovering the core.
Our algorithms are fast, simple to implement, and out-perform several
methods based on network core-periphery structure on 
various real-world datasets.
\end{abstract}

\maketitle

%!TEX root = found-graph-data-paper.tex

% core-fringe
% unknown fringe: provenance
% unknown fringe: counter-surveillance
% planted vc
% assumptions, results
% practical

\newcommand{\omt}[1]{}
\def\gap{\vskip 0.2in}
\def\scenebreak{{\gap \center{\rule[.1in]{4.0in}{.01in}} \gap}}

\section{Partially measured graphs, data provenance, and planted structure}

Datasets that take the form of graphs are ubiquitous throughout the sciences~\cite{Albert-2002-survey,Easley-2010-networks-book,Newman-2003-survey},
but the graph data that we work with is generally incomplete in certain
systematic ways
\cite{Gile-2010-RDS,Khabbazian-2017-RDS,Kim-2011-completion,Kossinets-2006-missing,Laumann-1989-boundary}.
Perhaps the most ubiquitous type of incompleteness comes from the
way in which graph data is generally measured:
we observe a set of nodes $C$ and record all the interactions
that they are involved in.
The result is a measured graph $G$ consisting of this
{\em core set} $C$ together with a a potentially larger set of additional 
{\em fringe nodes} --- the nodes outside of $C$ that some node in $C$ interacts with.
For example, in constructing a social network dataset, we might
study the employees of a company and record all of their
friendships~\cite{Romero-2016-stress}; from this information, we now have a graph that contains
all the employees together with all of their friends, including
friends who do not work for the company.
This latter group constitutes the set of fringe nodes in the graph.
The edge set of the graph $G$ reflects this construction process:
we can see all the edges that involve a core node, but if two nodes that
both belong to the additional fringe set have interacted, it
is invisible to us and hence not recorded in the data.

E-mail and other communication datasets typically look like this;
for example, the widely-studied 
Enron email graph ~\cite{Eppstein-2011-maximal,Koutra-2013-DeltaCon,Leskovec-2010-NCP,Seshadhri-2013-wedge} contains tens of thousands of nodes\footnote{\url{http://snap.stanford.edu/data/email-Enron.html}}\footnote{\url{http://konect.uni-koblenz.de/networks/enron}}; however,
this graph was constructed from the email inboxes of fewer than 150 employees~\cite{Klimt-2004-Enron}. 
The vast majority of the nodes in the graph, therefore, belong
to the fringe, and their direct communications are not part of the data.
The issue comes up in much larger network datasets as well.
For example, a telephone service provider has data on calls and messages
that its customers make both to each other and to non-customers; 
but it does not have data on communication between pairs of non-customers.
A massive social network may get some information about
the contacts of its users with a fringe set consisting
of people who are not on the system --- often including 
entire countries that do not participate in the platform --- but 
generally not about the interactions taking place in this fringe set.
And Internet measurements at the IP-layer of service providers provide
only a partial view of the Internet for 
similar reasons~\cite{Tsiatas-2013-spectral}.

This then is the form that much of our graph data takes
(depicted schematically in Figure \ref{fig:core-fringe-shown}):
the nodes are divided into a core set and a fringe set,
and we only see the edges that involve a member of the core set.
This means, in particular, that the core set is a {\em vertex cover}
of the underlying graph --- since a vertex cover, by definition, is
a set that is incident to all the edges.

\begin{figure}[b!]
\scenebreak
\begin{center}
\subfigure[\emph{Graph data built from a small core}]{
\includegraphics[width=.43\textwidth]{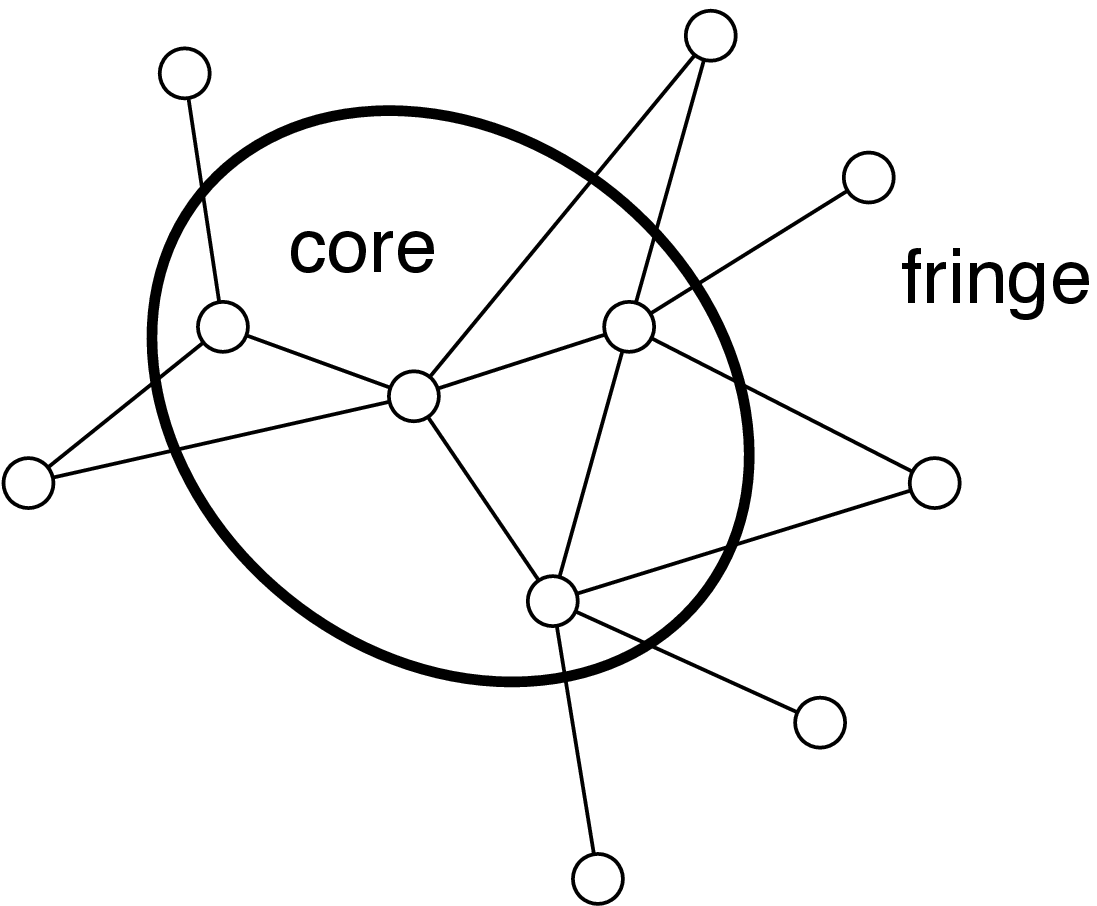}
 \label{fig:core-fringe-shown}
}
\hspace*{0.08\textwidth}
\subfigure[\emph{The dataset without the core labeled}]{
\includegraphics[width=.35\textwidth]{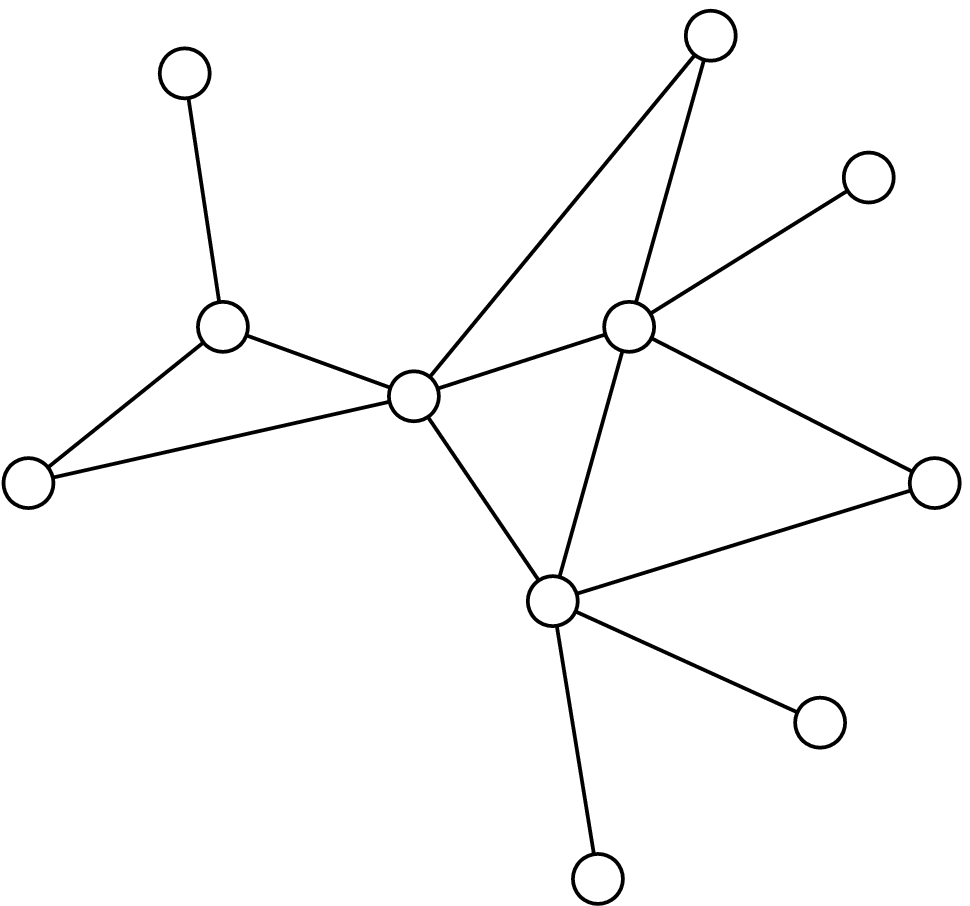}
 \label{fig:core-fringe-hidden}
}
\caption{
(a) Graph datasets are often constructed by recording all the interactions
involving a set of measured {\em core nodes}; the resulting data contains
these core nodes together with a potentially much larger 
{\em fringe}, consisting of all other nodes that had an interaction
with some member of the core.
We can see the links between members of the fringe and members of the core,
but links between pairs of nodes in the fringe are not recorded and hence
not visible in the data.
(b) It is important for interpreting the dataset that the metadata
describing which nodes belong to the core be preserved.
But in many cases, this metadata is not available --- either because
of challenges in data provenance that lead to the loss of the metadata,
or because the graph is ``found data'' obtained 
in a context such as counter-surveillance.
The challenge is then to determine how accurately we can recover the core, 
despite limited information about how the dataset was constructed.
Algorithmically, this leads to a {\em planted vertex cover} problem,
in which we are given a graph and are tasked with identifying a hidden vertex cover (the core).
\label{fig:core-fringe}
}
\end{center}
\end{figure}

In many cases, a graph dataset comes annotated with metadata
about which nodes belong to the core, and this is crucial
for correctly interpreting the data.
But there are a number of important contexts where this metadata
is not available, and we simply do not know which nodes constitute the core.
In other words, at some point, we have ``found data'' that we know has a core,
but the labels are missing.
One reason for this scenario is that over time, metadata becomes lost for a wide
range of reasons; this is a central underlying theme in 
data provenance, lineage, and preservation~\cite{Buneman-2001-provenance,Lynch-2008-data,Simmhan-2005-survey,Tan-2004-provenance}, and an issue 
that has become especially challenging with modern efforts in digitization~\cite{Kuny-1997-dark}
and the increasing size of data management efforts~\cite{Jagadish-2014-bigdata}.
For example, data is repeatedly shared and manipulated, URLs become defunct,
managers of datasets change jobs, and hard drives are decommissioned.
Natural scenarios embodying these forces proliferate;
consider for example an anonymized research dataset of telephone call records,
shared between a telecommunications company
and a university, which did not include metadata on which nodes were 
the customers and which were the fringe set of non-customers.
By the time the graduate student doing research on the dataset
comes to know that this metadata is missing and important for analysis, 
the researchers at the company who originally assembled the dataset
have left, and there is no easy way to reconstruct the metadata.

These issues arise in very similar forms in current research on
the process of counter-surveillance
\cite{Hier-2009-surveillance}.
Intelligence agencies may intercept data from adversaries conducting
surveillance and build a graph to determine which communications 
the adversaries were recording.
In different settings, 
activist groups may petition for the release of surveillance
data by governments, or infer it from other sources
\cite{Monahan-2010-counter-surveillance}.
In all these cases, the ``found data'' consists of a communication
graph in which an unknown core subset of the nodes was observed, 
and the remainder of the nodes in the graph (the fringe)
are there simply because
they communicated with someone in the core.
But there may generally not be any annotation in such situations
to distinguish the core from the fringe.
In this case, the core nodes are the compromised ones, and identifying 
the core from the data can help to
warn the vulnerable parties, hide future communications, 
or even disseminate misinformation.

\paragraph*{\bf Planted Vertex Covers.}
Here we study the problem of recovering the set of core nodes 
in found graph data, motivated by this range of settings in which
reconstructing an unknown core is a central question.
Algorithmically, the problem can be stated simply as a 
{\em planted vertex cover problem}: we are given a graph $G$
in which an adversary knows a specific vertex cover $C$;
we do not know the identity of $C$, but 
our goal is to output a set that is as close to $C$ as possible.
Here, the property of being ``close'' to $C$ corresponds to a
performance guarantee that we will formulate in several different ways:
we may, for example, want to output a set not much larger than $C$
that is guaranteed to completely contain it; or we may want to output
a small set that is guaranteed to have significant overlap with $C$.
A simple instance of the task is depicted in 
Figure~\ref{fig:core-fringe-hidden}, after the explicit labeling
of the core nodes has been removed from
Figure~\ref{fig:core-fringe-shown}.

Planted problems have become an active topic of study in recent years.
Generically they correspond to a style of problem in which some
hidden structure (like the vertex cover in our case) has been
``planted'' in a larger input, and the goal is to find the structure
in the given input.
Planted problems tend to be based on formal frameworks in which 
the input is generated by a highly structured probabilistic model.
Perhaps the two most heavily-studied instances are
the planted clique problem~\cite{Alon-1998-clique,Deshpande2014,Feige-2010-cliques,Meka-2015-SoS-clique}, in which a large clique is
added to an Erd\H{o}s-R\'enyi random graph; 
and the recovery problem for stochastic block models~\cite{Abbe-2017-community,Abbe-2016-exact,Abbe-2015-recovering,Bickel-2009-nonparameteric,Decelle-2011-SBM,Mossel-2014-belief}, in which a graph's edges are generated independently
at random, but with higher density inside communities than between.

It might seem essentially inevitable that planted problems should require
such strong probabilistic assumptions --- after all, how else could
an algorithm possibly guess which part of the graph corresponds to
the planted structure, if there are no assumptions on what the ``non-planted''
part of the graph looks like?

But the vertex cover problem turns out to be different, and it makes it
possible to solve what, surprisingly, can be described as a kind of 
``worst-case'' planted problem --- with
extremely limited assumptions, it is possible
to design algorithms capable of finding sets that are close to unknown vertex
covers in arbitrary graphs.
Specifically, we make only two assumptions about the input
(both of them necessary in some form, though relaxable):
that the planted vertex cover is inclusionwise minimal, and
that its size is upper-bounded by a known quantity $k$.
It is natural to think of $k$ as relatively small compared to the total
number of nodes $n$, in keeping with the fact that in much of
the measured graph data, the core is small relative to the fringe.
We draw on results from the theory of fixed-parameter tractability
to show that there is an algorithm operating on arbitrary graphs
that can output a set of
$f(k)$ nodes (independent of the size of the graph)
that is guaranteed to contain the planted vertex cover.
We obtain further results as well, including stronger bounds
when the size of the planted vertex cover is close to minimum;
when we can partially overlap the planted vertex cover without
fully containing it; and
when the graph is generated from a natural probabilistic model.

We pair these theoretical guarantees with computational experiments
in which we show the effectiveness of these methods on graph
datasets exhibiting this structure in practice.
Using the ingredients from our theoretical results, 
we develop a natural heuristic based on 
unions of minimal vertex covers, each obtained via pruning a
2-approximate maximal matching algorithm for minimum vertex covers with random
initialization. The entire algorithm is implemented in just 30 lines of Julia code (see \cref{fig:julia_umvc} for the complete implementation). 
Our algorithm provides 
superior empirical performance and superior running time
across a range of real datasets compared to
a number of competitive baseline algorithms.
Among these, we show improvements over a line of well-developed
heuristics for detecting {\em core-periphery structure} in 
graphs~\cite{Csermely-2013-CP,Holme-2005-CP,Rombach-2017-CP,Zhang-2015-SBM-CP} ---
a sociological notion related to our concerns here, in which
a graph has a dense core and a sparser periphery, generally for
reasons of differential status rather than measurement 
effects.\footnote{In the terminology of core-periphery models,
our fringe set with no internal edges corresponds to a 
``zero block,'' in which the periphery nodes
only connect via paths through the core~\cite{Borgatti-2000-CP,Breiger-1981-structures,Burt-1976-Positions}.  The methods associated with this concept
in earlier work, however, do not yield our theoretical guarantees
nor the practical performance of the heuristics we develop.}

%!TEX root = found-graph-data-paper.tex

\section{Theoretical methodology for partial recovery or tight containment}\label{sec:theory}

We begin by formalizing the recovery problem.
Suppose there is a large universe of nodes $U$ that interact
via communication, friendship, or some other mechanism. 
We choose a {\em core} subset $C \subseteq U$ of these nodes and measure
all the pairwise interactions that involve at least one node in $C$.
We represent our measurement on $C$ by a graph $G = (V, E)$, where
$V \subseteq U$ is all nodes that belong to $C$ or 
participate in a pairwise link
with at least one node in $C$, and $E$ is the set of all such links. 
The set of nodes in $V - C$ will be called the {\em fringe} of the graph.
We will ignore
directionality and self-loops so that $G$ is a simple, undirected graph. 
Note that under this construction, $C$ is a vertex cover of $G$.

\subsection{Finding a planted vertex cover}

Now, suppose we are shown the graph $G$ and are tasked
with finding the core $C$.
Can we say anything non-trivial in answer to this question?
In the absence of any other information, it could be
that $C = V$, so we first make the assumption that we are given 
a bound $k$ on the size of $C$, where we think of $k$ as
small relative to the size of $V$.
With this extra piece of information, we can ask if it is
possible to obtain a small set that is guaranteed
to contain $C$. We state this formally as follows.
\begin{question}
For some function $f$, can we find a set $D$
of size at most $f(k)$ (independent of the size of $V$)
that is guaranteed to contain the planted vertex cover $C$?
\end{question}

The answer to this question is ``no''.
For example,
let $k = 2$ and let $G$ be a star graph with $n > 3$ nodes $v_1, ..., v_n$
and edges $(v_1, v_i)$ for each $i > 1$.
The two endpoints of any edge in the graph form a vertex cover
of size $k = 2$, but $C$ could conceivably be any edge.
Thus, under these constraints, the only set guaranteed to contain $C$
is the entire node set $V$.

This negative example uses the fact that once we put $v_1$ into a 2-node vertex
cover for the star graph, the other node can be arbitrary, since it is
superfluous.  The example thus suggests that a much more reasonable version
of the question is to ask about minimal vertex covers;
formally, $C$ is a {\em minimal vertex cover} if for all $v \in C$,
the set $C - \{v\}$ is not a vertex cover.
(We contrast this with the notion of a {\em minimum} vertex cover ---
a definition we also use below --- which is a vertex cover whose
size is minimum among all vertex covers for the given graph.)
Minimality is a natural definition with respect to our original
motivating application as well, where it would be reasonable to assume that
the set of measured nodes was non-redundant, in the sense that omitting
a node from $C$ would cause at least one edge to be lost from the
measured communication pattern.
This would imply that $C$ is a minimal vertex cover.
With this in mind, we ask the following adaptation of Question 1:

\begin{question}
If $C$ is a minimal planted vertex cover, can we find a set $D$
of size at most $f(k)$ (independent of $\lvert V \rvert$) that is
guaranteed to contain $C$?
\end{question}

Interestingly, the answer to this question is ``yes,'' for arbitrary graphs.
We derive this as a consequence of a result due to Damaschke 
\cite{Damaschke-2006-parameterized,Damaschke-2009-unions}
in the theory of fixed-parameter tractability.
While this theory is generally motivated by the design of fast
algorithms, it is also a source of important structural results about graphs.
The structural result we use here is the following.

\begin{lemma}[{{\cite{Damaschke-2006-parameterized,Damaschke-2009-unions}}}]\label{lem:bounds}
  Consider a graph $G$ with a \emph{minimum} vertex cover size $k^*$.
  Let $U(k)$ be the union of all minimal vertex covers of size at most $k$. Then
\begin{itemize}
\item[(a)] $\lvert U(k) \rvert \le (k + 1)^2 / 4 + k$ and is asymptotically tight~\cite[Theorem 3]{Damaschke-2006-parameterized}
\item[(b)] $\lvert U(k) \rvert \le (k - k^* + 2)k^*$ and is tight~\cite[Theorem 12]{Damaschke-2009-unions}
\end{itemize}
\end{lemma}

For part (a) of this lemma, there is an appealingly direct proof
that gives the $O(k^2)$ asymptotic bound, using
the following {\em kernalization} technique from the theory
of fixed-parameter 
tractability~\cite{Buss-1993-nondeterminism,Downey-Fellows-2012-book}.
The proof begins from the following observation:
\begin{observation}\label{obs:degree_include}
Any node with degree strictly greater than $\sizeC$ must be in $C$.
\end{observation}
The observation follows simply from the fact that if a node is omitted
from $C$, then all of its neighbors must belong to $C$.
Thus, if $S$ is the set of all nodes in $G$ with degree greater than $k$,
then $S$ is contained in every vertex cover of size at most $k$;
hence \cref{obs:degree_include} implies that if $U(k)$ is non-empty,
we must have $|S| \leq k$ and $S \subseteq U(k)$.
Now $G - S$ is a graph with maximum degree $k$ and a vertex cover of
size $\leq k$, so it has at most $O(k^2)$ edges; let $T$ be the set
of all nodes incident to at least one of these edges.
Any node not in $S \cup T$ is isolated in $G - S$ and hence not part
of any minimal vertex cover of size $\leq k$; 
therefore $U(k) \subseteq S \cup T$, and so $|U(k)| = O(k^2)$.

The following theorem, giving a positive answer to Question 2, 
is thus a corollary of \cref{lem:bounds}(a) obtained by setting $D = U(k)$.

\begin{theorem}\label{thm:containment}
If $C$ is a minimal planted vertex cover with $\lvert C \rvert \le k$, 
then we can find a set $D$ of size $O(k^2)$ that is guaranteed to contain $C$.
\end{theorem}

To see why $O(k^2)$ is a tight bound, consider a graph $G$
consisting of the disjoint union of $k/2$ stars each with $1 + k/2$ leaves.
Any set consisting of the centers of all but one of the stars, and
the leaves of the remaining star, is a minimal vertex cover of size $k$.
But this means that every node in $G$ could potentially belong
to the planted vertex cover $C$, and so the only acceptable answer is
to output the full node set $V$.
Since $V$ has size $\Omega(k^2)$, the bound follows.

There are algorithms to compute $U(k)$ that run in time
exponential in $k$ but polynomial 
in the number of nodes for fixed $k$~\cite{Damaschke-2009-unions}. 
For the datasets we consider in this paper,
these algorithms are impractical, despite the polynomial 
dependence on the number of nodes. 
However, we will use the results in this section as the basic
ingredients for an algorithm, developed in \cref{sec:experiments},
that works extremely well in practice.

\subsection{Non-minimal vertex covers}

A natural next question is whether we can say anything positive
when the planted vertex cover $C$ is not minimal.
In particular, if $C$ is not minimal, can we still ensure
that some parts of it must be contained in $U(\sizeC)$?
The following propositions show that if a node $u \in C$
links to a node $v$ that is outside $C$, or that 
is deeply contained in $C$ (with $v$ and its neighbors all in $C$),
then $u$ must belong to $U(\sizeC)$.

\begin{proposition}\label{prop:cp_in_union}
If $u \in C$ and there is an edge $(u, v)$ to a fringe node $v \notin C$, then 
$u \in U(\sizeC)$.
\end{proposition}
\begin{proof}
Consider the following iterative procedure for ``pruning'' the set $C$:
we repeatedly check whether there is a node
$w$ such that $C - \{w\}$ is still a vertex cover; and if so,
we choose such a $w$ and delete it from $C$.
When this process terminates, we have a minimal vertex cover
$C' \subseteq C$; and since $|C'| \leq \sizeC$, we must have
$C' \subseteq U(\sizeC)$.
But in this iterative process we cannot delete $u$, since 
$(u,v)$ is an edge and $v \not\in C$.
Thus $u \in C'$, and hence $u \in U(\sizeC)$.
\end{proof}

Next, let us say that a node $v$ belongs to the {\em interior} 
of the vertex cover $C$ if $v$ and all the neighbors of $v$
belong to $C$.
We now have the following result.

\begin{proposition}\label{prop:ni_in_union}
If $u \in C$ and there is an edge $(u, v)$ to a node $v$
in the interior of $C$, then 
$u \in U(\sizeC)$.
\end{proposition}
\begin{proof}
Let $u$ and $v$ be nodes as described in the statement of the proposition.
Since all of $v$'s neighbors are in $C$, it follows
that $C_0 = C - \{v\}$ is a vertex cover.
We now proceed as in the previous proof: we iteratively
delete nodes from $C_0$ as long as we can preserve the
vertex cover property. When this process terminates,
we have a minimal vertex cover $C' \subseteq C_0$,
and since $|C'| \leq \sizeC$, we must have $C' \subseteq U(\sizeC)$.
Now, $u$ could not have been deleted during this process,
because $(u,v)$ is an edge and $v \not\in C_0$.
Thus $u \in C'$, and hence $u \in U(\sizeC)$.
\end{proof}

%As a direct corollary of \cref{prop:ni_in_union},
%we have the following further result.
%
%\begin{proposition}\label{prop:n2_in_union}
%Let $u$ be a non-isolated node, and suppose the
%2-hop neighborhood of $u$ --- the set consisting of $u$, all of $u$'s neighbors, and all of $u$'s neighbors' neighbors --- is a subset of $C$. 
%Then $u \in U(\sizeC)$.
%\end{proposition}

Even with these results, we can find instances where
an arbitrarily small fraction of the 
nodes in a non-minimal planted vertex cover
$C$ may be contained in $U(\sizeC)$.
To see this, consider a star with center node $u$ and 
$k+1$ leaves; and let $C$ consist of $u$ together with any
$k-1$ of the leaves.
We observe that the single-node set $\{u\}$ is the only
minimal vertex cover of size at most $k$, and hence
$\lvert U(\sizeC) \rvert / \sizeC = 1 / k$.
Note how in this example, all the other nodes of $C$ fail to
satisfy the hypotheses of Proposition \ref{prop:cp_in_union} or \ref{prop:ni_in_union}.
However, we will see in \cref{sec:experiments} that in all
of the real-world network settings we consider,
these three propositions can be used to
show that most of $C$ is indeed contained in $U(\sizeC)$.

Our bad example consisting of a star also has the property that 
the planted vertex $C$ is much larger than the size of a minimum vertex cover.
We next consider the case in which $C$ may be non-minimal, but
is within a constant multiplicative
factor of this minimum size $k^*$. 
In this case, we will show how to
find small sets guaranteed to intersect a constant fraction of
the nodes in $C$.

\subsection{Maximal matching 2-approximation to minimum vertex cover and intersecting the core}\label{sec:matching}

%\begin{wrapfigure}{R}{0.5\textwidth}
%\begin{minipage}{1\linewidth}
\begin{algorithm}[H]
  \DontPrintSemicolon
  \KwIn{Graph $G = (V, E)$}
  \KwOut{Vertex cover $M$ with $\lvert M \rvert \le 2k^*$}
  \caption{Maximal matching 2-approximation for minimum vertex cover}
  $M \leftarrow \emptyset$\;
  \For{$e = (u, v) \in E$}{
    \lIf{$u \notin M$ and $v \notin M$}{$M \leftarrow M \cup \{u, v\}$}
  }
\label{alg:matching}
\end{algorithm}
%\end{minipage}
%\end{wrapfigure}

A basic building block for our theory in this section and 
the algorithms we develop later is the classic maximal
matching 2-approximation to minimum vertex cover (\cref{alg:matching}, above). 
The algorithm greedily builds a maximal matching $M$ by processing each edge 
$e = (u, v)$ of the graph and adding $u$ and $v$
to $M$ if neither endpoint is already in $M$.
Upon termination, $M$ is both a maximal matching and a vertex cover:
it is maximal because if we could add another edge $e$, then it would have been added when we processed edge $e$;
and it is a vertex cover because if both endpoints of an edge $e$ are not in the matching, then
$e$ would have been added to the matching when it was processed.
Since any vertex cover must contain at least one endpoint from each edge in the matching, we have $k^* \geq \lvert M \rvert / 2$;
or, equivalently, $\lvert M \rvert \le 2k^*$,
where $k^*$ is the minimum vertex cover size of $G$.
We note that the output $M$ of \cref{alg:matching} may not be a \emph{minimal}
vertex cover. 
However, we can iteratively
prune nodes from $M$ to make it minimal, which we will 
do for our recovery algorithm described in \cref{sec:experiments}. 
For the theory in this section, though, we assume no such pruning.

The following proposition shows that any vertex cover whose size is bounded by a constant multiplicative
factor of the minimum vertex cover size must intersect the output of \cref{alg:matching} in a constant fraction
of its nodes.
\begin{lemma}
  Let $B$ be any vertex cover of size $\lvert B \rvert \le bk^*$ for some constant $b$.
  Then any set $M$ produced by \cref{alg:matching} satisfies
  $\lvert M \cap B \rvert \ge \frac{1}{2b}\lvert B \rvert$.
\end{lemma}
\begin{proof}
  The maximal matching consists of $h$ edges satisfying 
  $h = |M|/2 \le k^*$.
  Since $B$ is a vertex cover, it must contain at least one endpoint of each
  of the $h$ edges in $M$. Hence, $\lvert M \cap B \rvert \ge h \ge k^*/2 \ge \frac{1}{2b}\lvert B \rvert$.
\end{proof}  

A corollary is that if our planted cover $C$ is 
relatively small in the sense that it is close to
the minimum vertex cover size, 
then \cref{alg:matching} must partially recover $C$.
We write this as follows.
\begin{corollary}\label{prop:core_overlap}
If the planted vertex cover $C$ has size $\sizeC \le ck^*$, 
then \cref{alg:matching} produces a set $M$ of size $\leq 2k^*$
that intersects at least a $1/(2c)$ fraction of the nodes in $C$.
\end{corollary}

An important property of \cref{alg:matching} that will be useful for our algorithm design later
in the paper is that the algorithm's guarantees hold
regardless of the order in which the edges are processed. 
Furthermore, two matchings produced by the algorithm using
two different orderings of the edges
must share a constant fraction of nodes, 
as formalized in the following corollary.
\begin{corollary}\label{cor:matching_overlap}
Any two sets $S_1$ and $S_2$ obtained from \cref{alg:matching} 
(with possible different orders in the processing of the edges) satisfy
$\lvert S_1 \cap S_2 \rvert \ge \frac{1}{4}\max(\lvert S_1 \rvert, \lvert S_2 \rvert)$.
\end{corollary}

Thus far, everything has applied without any
assumptions on the graph itself, other than the fact
that it contains a vertex cover of size at most $k$.
In the following section, we show how different ways of
assuming structure on the graph can yield stronger 
guarantees.

\subsection{Improving results with known graph structure}

We now examine ways to strengthen 
our theoretical guarantees by assuming some structure
on $C$ and $G$. 
For a simple example, 
if $C$ induces a clique in $G$, then the subgraph on $C$
has a minimum vertex cover of size $\sizeC - 1$; 
in this case the second bound in \cref{lem:bounds} reduces to 
$\lvert U(\sizeC) \rvert \le 3\sizeC$.
Furthermore, in this case, \cref{obs:degree_include} 
would say that any node in $C$ that
has an edge to a fringe node outside of $C$ is immediately identifiable from its degree (if we knew the size of $C$).
Below, we consider how to make use of random structure or bounds on the minimum vertex
cover size $k^*$ obtained through computation with \cref{alg:matching} to 
strengthen our theoretical guarantees.

\xhdr{Stochastic block model.}
One common structural assumption is 
that edges are generated independently at random.
The stochastic block model (SBM) is a common generative model for this
idealized setting~\cite{Holland-1983-SBM}. In our case, we use a 
2-block SBM, where one block is the planted vertex cover $C$ and the other
block is the remaining fringe nodes $F = V - C$. The SBM
provides a single probability of an edge forming within a block and between
blocks. 
For our purposes, $C$ is a vertex cover, so we assume that the
probability of an edge between nodes in $F$ is 0. For notation, we will
say that the probability of an edge between nodes in $C$ is $p$ and the
probability of an edge between a node in $C$ and a node in $F$ is $q$.
We make no assumption on the relative values of $p$ and $q$.
SBMs have previously been used to model core-periphery 
structures in networks~\cite{Zhang-2015-SBM-CP},
although this prior work assumes that $p > q$ and that there is a probability $r < q$
of nodes in the periphery forming an edge.
(In \cref{sec:experiments}, we compare against the belief
propagation algorithm developed in this prior work.)

Our technical result here combines the second bound in \cref{lem:bounds}
with the well-known lower bounds on independent set size in
Erd\H{o}s-R\'enyi graphs (in the SBM, $C$ is an Erd\H{o}s-R\'enyi graph
with edge probability $p$).

\begin{lemma}\label{lem:sbm}
With probability at least $1 - \sizeC^{-3\ln n / (2p)}$), 
the union of minimal vertex covers of size at most
$\sizeC$ contains at most  $\sizeC(3\ln \sizeC / p + 3)$ nodes.
\end{lemma}
\begin{proof}
It is straightforward to show that the independence number $\alpha$ of $C$ is
less than $(3\ln \sizeC) / p + 1$ with probability at least 
$1 - \sizeC^{-3\ln \sizeC / (2p)}$~\cite{Spielman-2010-ER}.
The minimum vertex cover size of the first block is then 
$k^* = \sizeC - \alpha \ge \sizeC - (3\ln \sizeC) / p - 1$ 
with the same probability.
Plugging into \cref{lem:bounds}(b) gives the result.
\end{proof}
We could improve the constants in the above statement with more technical
results on the independence number in Erd\H{o}s-R\'enyi 
graphs~\cite{Dani-2011-independence-number,Frieze-1990-independence-number}.
However, our point here is simply that the SBM provides substantial structure.
The following theorem further represents this idea.

\begin{theorem}\label{thm:sbm}
Let $C$ be a planted vertex cover in our SBM model, where
we know that $\sizeC = k$.  Let $p$ and $q$ be constants, and let
the number of nodes in the stochastic block model be $ck$ for some constant $c \ge 1$.
Then with high probability as a function of $k$, 
there is a set $D$ of size $O(k\log k)$
that is guaranteed to contain $C$.
\end{theorem}
\begin{proof}
By \cref{lem:sbm}, we know that $\lvert U(k) \rvert$ is $O(k \log k)$ with high probability.
Now, for any node $v \in C$, the probability that it links to at
least one node outside of $C$ is $1 - (1 - q)^{(c-1)k}$.
Taking the union bound over all nodes in $C$ shows that with
high probability in $k$, each node in $C$ has at least one
edge to a node outside $C$.
In this case, \cref{prop:cp_in_union} implies that $C \subseteq U(k)$,
so by computing $U(k)$, we contain $C$ with high probability.
\end{proof}

%We can also refine the theory of \cref{prop:cover} for the case of the SBM.
%\begin{proposition}
%Let $M$ be any matching output by \cref{alg:matching}. Then with high probability,
%\[
%\frac{\lvert C \cap M \rvert}{\sizeC} \ge 1 - O\left(\frac{\log \sizeC}{\sizeC}\right)
%\]
%\end{proposition}
%\begin{proof}
%Let $S = M \backslash C$ be the set of nodes in $C$ that do not appear in the matching.
%Since the minimum vertex cover must be 
%By \cref{prop:cover}, the minimum vertex cover of the graph is of size at least
%$\sizeC - O(\log \sizeC)$. Thus, $\lvert S \rvert \le O(\log \sizeC)$.
%\end{proof}

\xhdr{Bounds on the minimum vertex cover size.}
While it may sometimes be impractical to compute the \emph{minimum} vertex cover size $k^*$,
the second bound of \cref{lem:bounds} may still be used if we can bound $k^*$ from above
and below. Specifically, given a lower bound $l$ and an upper bound $u$ on $k^*$,
$\lvert U(k) \rvert \le (k - k^* + 2)k^* \le (k - l + 2)u$.
Here, we have a scenario in which we want the lower
bound $l$ to be as large 
as possible and the upper bound $u$ 
to be as small as possible.

A cheap way of finding such bounds is to use
the greedy maximal matching approximation algorithm (\cref{alg:matching}).
We can run the approximation algorithm $N$ times, processing the edges in different (random) orders, producing sets $S_1, \ldots, S_N$ of sizes $s_1, \ldots, s_N$. We can then set the largest lower bound obtained
to be $l = \max_{j} s_j / 2$.
As noted in \cref{sec:matching}, the sets $S_j$ are not guaranteed to be minimal. 
We can post-process them to be minimal, producing sets $S'_1, \ldots, S'_N$
of sizes $s'_1, \ldots, s'_N$. Setting $u = \min_{j} s'_j$ gives us the smallest upper bound.
Through the lens of this procedure, the search for a large lower bound makes sense. 
If we can find a lower bound $l$ and an upper bound $u$ such that $l = 2u$, 
then we know that $k^* = l$. In other words, larger lower bounds, combined 
with the 2-approximation, are giving us more information on the interval containing $k^*$.

\begin{figure}[tb]
\newcommand{\skipsize}{-6pt}
  \centering
 \includegraphics[width=0.325\columnwidth]{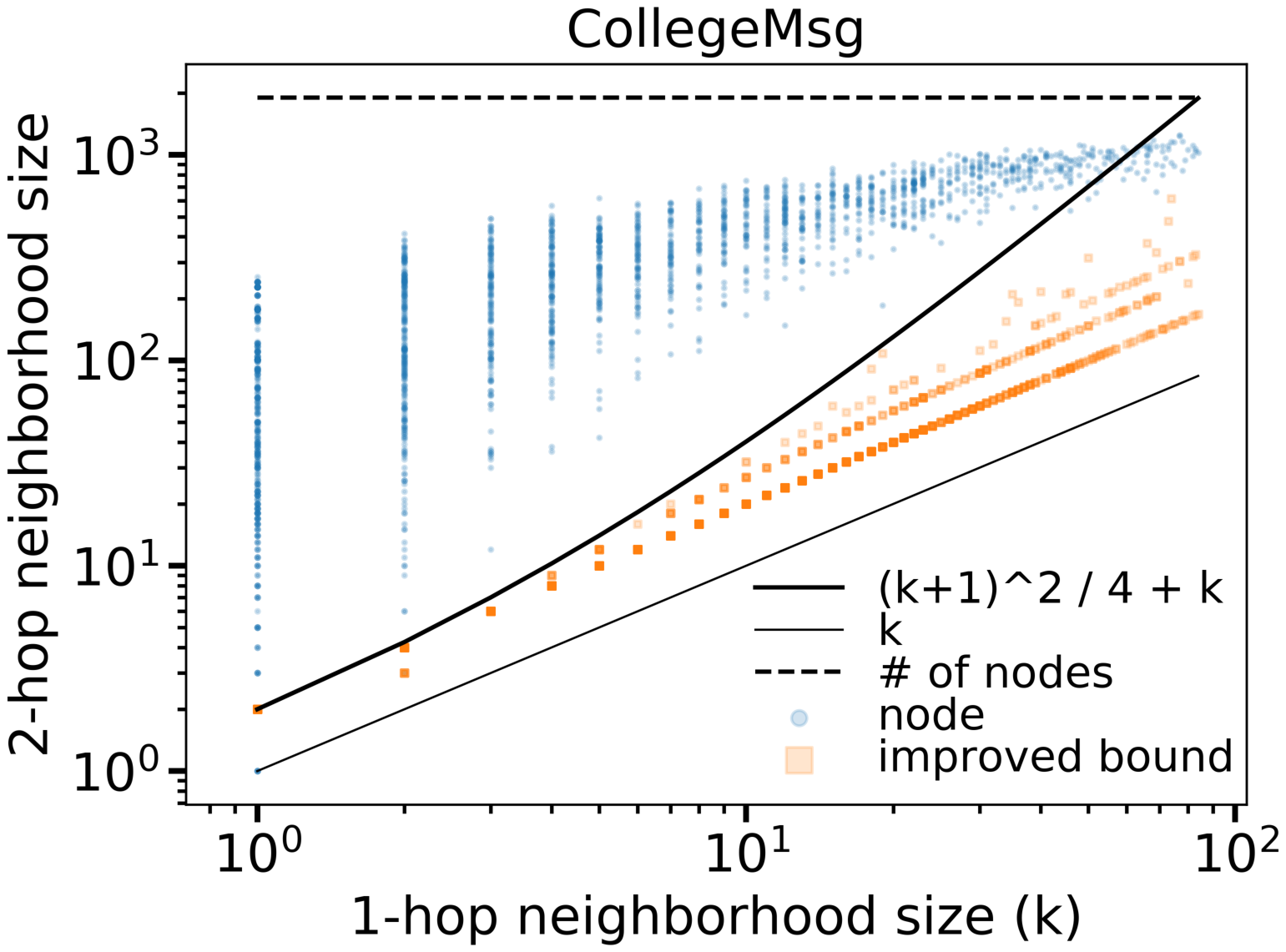}
 \includegraphics[width=0.325\columnwidth]{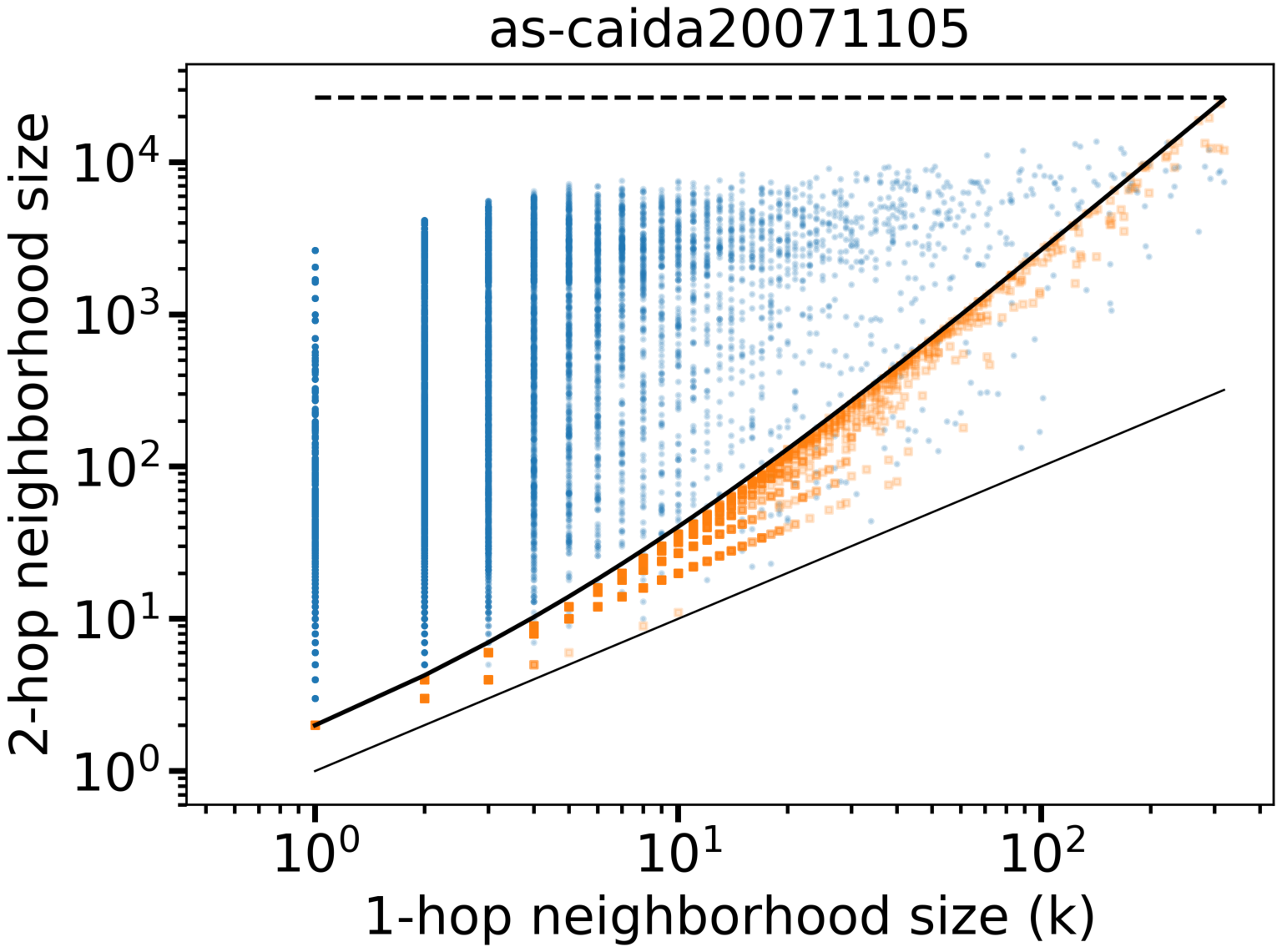}
 \includegraphics[width=0.325\columnwidth]{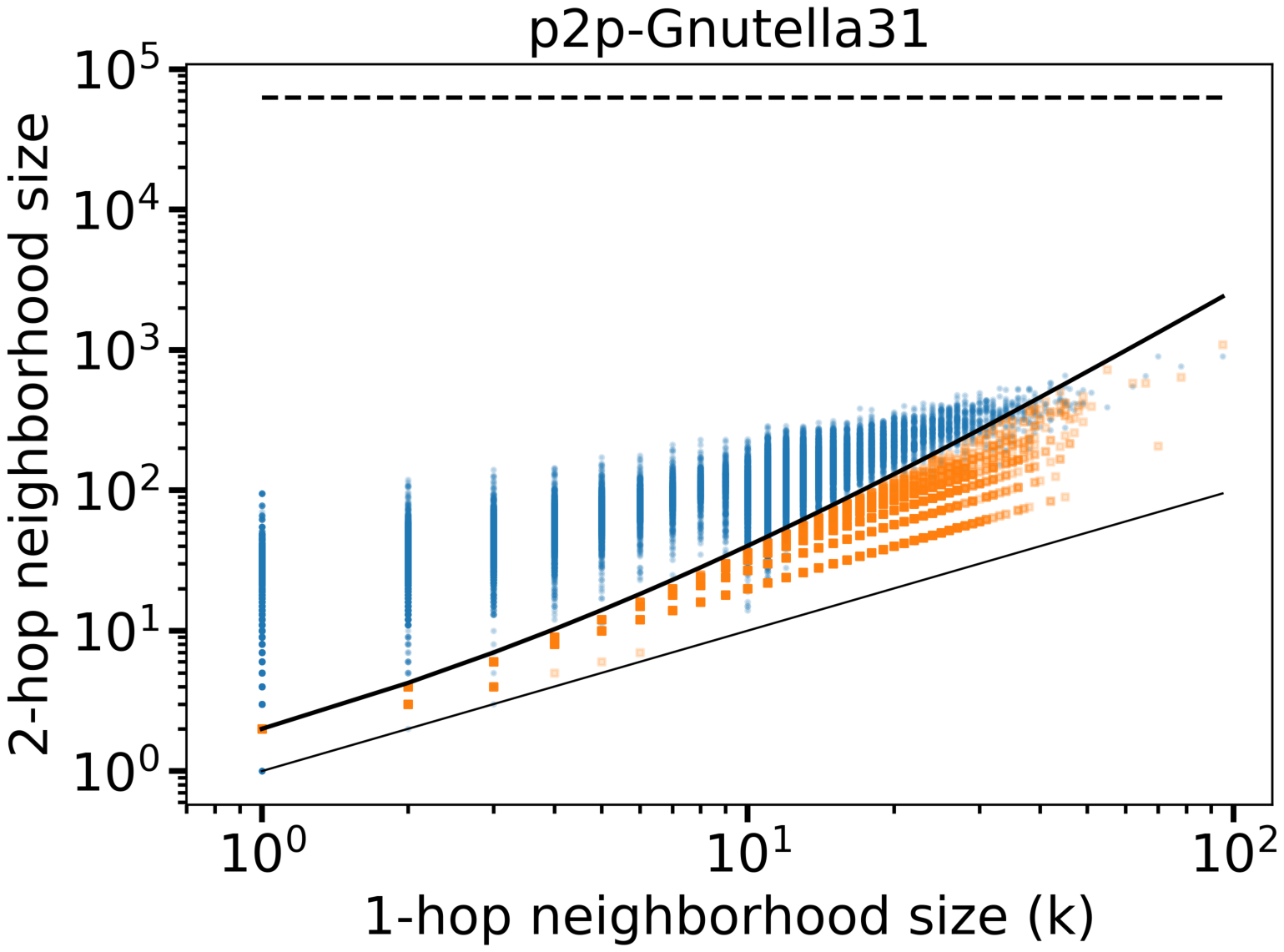}
 \caption{Improvements in the second bound from \cref{lem:bounds} by bounding the
 minimum vertex cover size. The planted vertex covers are 1-hop neighborhoods of a
 node, which cover the 2-hop neighborhoods (the center node is excluded).
 Blue dots in the scatter plot show the relative sizes of the planted cover and the total number
 of vertices in the subgraph. Blue dots above the thick black curve represent
 non-trivial bounds from the first bound in \cref{lem:bounds} 
 (i.e., the size of the union of covers is less than the size of the 2-hop neighborhood itself).
 The orange squares show improvements in the bound by approximating the minimum 
 vertex cover size. In many cases, this substantially improves the bound. 
 For the CollegeMsg dataset, the bounds appear to be linear instead of quadratic.}
  \label{fig:bound_improve}
\end{figure}

Before getting to our main computational results in the next section, 
we begin by evaluating this particular bounding methodology on three datasets: 
(i) a network of private messages sent on an online social network at the University of California, Irvine (CollegeMsg;~\cite{Panzarasa-2009-patterns}),
(ii) an autonomous systems graph derived from a RouteViews BGP table snapshot in November, 2007 (as-caida20071105;~\cite{Leskovec-2005-graphs}),
(iii) a snapshot of the Gnutella peer-to-peer file sharing network from August 2002 (p2p-Gnutella31;~\cite{Matei-2002-Gnutella}).
For each network, we construct planted covers by considering the 1-hop neighborhood of nodes $u$ as a cover for the 2-hop neighborhood of the same node (and we remove the node $u$ from both of these sets). This provides a collection of planted vertex covers in communication-like datasets.
We set $N = 20$ and use the above procedure to employ the second bound in \cref{lem:bounds}
without computing any minimum vertex covers.
\Cref{fig:bound_improve} summarizes the results. We observe that in most cases, the first
bound in \cref{lem:bounds} is non-trivial (i.e., is smaller than the size of the 2-hop neighborhood) and that the approximations can substantially improve the bound. In the case of the CollegeMsg dataset, the upper bounds appear approximately linear in the cover size.

%!TEX root = found-graph-data-paper.tex

\section{Recovery performance on datasets with ``real'' planted vertex cores}\label{sec:experiments}

\begin{table}[tb]
\setlength{\tabcolsep}{2pt}
\centering
\caption{Basic summary statistics of our graph datasets: 
the
number of nodes ($n$), 
number of edges ($m$), 
total time spanned by the dataset, 
size of the planted vertex cover ($\lvert C \rvert$), 
minimum vertex cover size ($k^*$), 
bounds from \cref{lem:bounds} given as a fraction of the total number of nodes (capped at 1, the trivial bound), 
fraction of nodes in $u \in C$ that are in an edge with $v \notin C$, and 
fraction of nodes in $C$ with a neighbor $v$ in the interior of $C$ ($v$ and all of its neighbors are in $C$).
The nodes in these last two categories are guaranteed to be in the union of
all minimal vertex covers of size at most $\lvert C \rvert$ by
\cref{prop:cp_in_union,prop:ni_in_union}.
Across all datasets, most nodes in $C$ fall into these categories.
}
\begin{tabular}{l c c c c c c c c c c c c}
\toprule
Dataset & n & m & time span & $\lvert C \rvert$ & $k^*$ & Bound 1 & Bound 2 & frac. $C$ w/ & frac. $C$ w/  \\
             &    &     & (days) & & & & & edge outside $C$ & interior neighbor \\
\midrule
email-Enron  & 18.6k & 43.2k & 1.50k & 146   & 146   & 0.30 & 0.02 & 0.99 & 0.00 \\
email-W3C    & 20.1k & 31.9k & 7.52k & 1.99k & 1.11k & 1.00 & 1.00 & 0.76 & 0.06 \\
email-Eu       & 202k  & 320k  & 804   & 1.22k & 1.18k & 1.00 & 0.26 & 0.99 & 0.00 \\
call-Reality    & 9.02k & 10.6k & 543   & 90    & 82    & 0.24 & 0.09 & 0.90 & 0.01 \\
text-Reality   & 1.18k & 1.95k & 478   & 84    & 80    & 1.00 & 0.41 & 0.88 & 0.00 \\
\bottomrule
\end{tabular}
\label{tab:summary_stats}
\end{table}

We now study how well we can recover planted vertex covers, where the vertex
cover corresponds to a core set arising from the type of measurement
process described in the introduction.
We use five datasets for this purpose:
\begin{enumerate}
\item email-Enron~\cite{Klimt-2004-Enron}: 
This is the dataset discussed in the introduction, where the core is the
set of email addresses for which the inboxes were released as part of the
investigation by the Federal Energy Regulatory Commission. 
Nodes are email addresses and there is an edge between two addresses if an email
was sent between them.
\item email-W3C~\cite{Craswell-2005-TREC,Oard-2006-TREC,Wu-2006-exploratory}:
This is a dataset of email threads crawled from W3C mailing lists. We consider
the ``core'' to be the nodes with a w3.org domain in the email address. There is
an edge between two email addresses if an email was sent between them.
\item email-Eu~\cite{Leskovec-2007-densification,Yin-2017-local}:
This dataset consists of email communication involving members of a European
research institution. The core represents the members of the research institution.
\item call-Reality~\cite{Eagle-2005-Reality}:
This dataset consists of phone calls made and received by a set of students
and faculty at the MIT Media Laboratory or MIT Sloan business school as part
of the reality mining project. 
These students and faculty constitute the core. 
There is an edge between any two phone numbers between which a call was made.
\item text-Reality~\cite{Eagle-2005-Reality}:
This dataset has the same core nodes as the phone-Reality dataset but
edges are formed via SMS text communications instead of phone calls.
\end{enumerate}
Each dataset has timestamps associated with the nodes, and we will evaluate
how well we can recover the core as the networks evolve.
\Cref{tab:summary_stats} provides some basic summary statistics of the datasets.
The table includes the minimum vertex cover size (computed using
Gurobi's linear integer program solver), which lets us evaluate the second bound
of \cref{lem:bounds}. We also computed the fraction of nodes that are guaranteed
to be in $U(\lvert C \rvert)$ by \cref{prop:cp_in_union,prop:ni_in_union} and find
that 82\%--99\% of the nodes fit these guarantees, depending on the dataset.

\subsection{The union of minimal vertex covers algorithm and recovery performance}

We now study how well we can recover the planted vertex cover 
consisting of the core $C$.
All of the methods we use provide an ordering on the nodes, often through some score function on the nodes in the graph.
We then evaluate recovery on two criteria: (i) precision at core size (i.e., the fraction of nodes in the top $\lvert C \rvert$ of the
ordering that are actually in $C$) and (ii) area under the precision recall curve (this metric is more appropriate than area under
the receiver operating characteristic curve when there is class imbalance~\cite{Davis-2006-AUPRC}, which is the case here).

%!TEX root = found-graph-data-paper.tex

\lstdefinelanguage{Julia}%    
{morekeywords={abstract,break,case,catch,const,continue,do,else,elseif,end,export,false,for,function,immutable,import,importall,if,in,macro,module,otherwise,quote,return,switch,true,try,type,typealias,using,while},%
sensitive=true,%
alsoother={$},%
morecomment=[l]\#,%
morecomment=[n]{\#=}{=\#},%
morestring=[s]{"}{"},%
morestring=[m]{'}{'},% 
}[keywords,comments,strings]%
\lstset{%
language         = Julia,
basicstyle       = \footnotesize \ttfamily,
keywordstyle     = \bfseries\color{blue},
stringstyle      = \color{magenta},
commentstyle     = \color{ForestGreen},                                                                                                                                         
showstringspaces = false,
numbers=left,                                                                                                                                                       
}
\begin{figure}[tb]
\centering
    \scalebox{0.75}{
\begin{lstlisting}
function UMVC(A::SparseMatrixCSC{Int64,Int64}, ncovers=300)
  edgevec = filter(e -> e[1] < e[2],
                   collect(zip(findnz(A)[1:2]...)))
  umvc = zeros(Int64, size(A,1))
  for _ in 1:ncovers
    # Run 2-approximation with random edge ordering
    cover = zeros(Int64, size(A,1))
    edge_queue = shuffle(edgevec)
    while length(edge_queue) > 0
      i, j = pop!(edge_queue)
      if cover[[i,j]] == [0,0]; cover[[i,j]] = [1,1]; end
    end
    # Reduce to a minimal cover
    while true
      reduced = false
      for c in shuffle(find(cover .== 1))
        nbrs = find(A[:,c])
        if sum(cover[nbrs]) == length(nbrs)
          cover[c] = 0
          reduced = true
        end
      end
      if !reduced; break; end
    end
    umvc[cover .== 1] = 1
  end
  d = vec(sum(A,2))  # degrees of nodes
  return [sort(find(umvc .== 1), by=v->d[v], rev=true);
          sort(find(umvc .== 0), by=v->d[v], rev=true)]
end
\end{lstlisting}}
    \caption{Complete implementation of our union of minimal vertex covers (UMVC) algorithm in 
    30 lines of Julia code. The algorithm repeatedly runs the standard maximal matching
    2-approximation algorithm for minimum vertex cover (\cref{alg:matching}; lines 6--12) and 
    reduces each cover to a minimal one (lines 13--24). The union of covers is ranked first 
    in the ordering (sorted by degree; line 28). The remaining nodes are then sorted by 
    degree (line 29). The code is available as a Gist at
\url{https://gist.github.com/arbenson/c1251f0851fdfc97021b9a0a37d4fb69}.}
    \label{fig:julia_umvc}
\end{figure}

\xhdr{Proposed algorithm: union of minimal vertex covers (UMVC).}
Our proposed algorithm, which we call the union of minimal covers (UMVC), repeatedly
finds minimal vertex covers and takes their union. 
The nodes in this union are ordered by degree and the remaining nodes (not appearing in any minimal core) are ordered by degree. 
The minimal covers are constructed by first finding a 2-approximate
solution to the \emph{minimum} vertex cover problem using the standard greedy algorithm
(\cref{alg:matching}) and then pruning the resulting cover to be minimal. 
We randomly order the edges for processing by the approximation
algorithm in order to capture different minimal covers. 
The algorithm is incredibly simple---\cref{fig:julia_umvc}
shows a complete implementation of the method in just 30 lines of Julia code. 
In our experiments, we use 300 minimal vertex covers.

Our algorithm is motivated by the theory in \cref{sec:theory} in several ways.
First, we expect that most of $C$ will lie in the union of all minimal vertex covers of
size at most $\lvert C \rvert$ by \cref{prop:cp_in_union,prop:ni_in_union,thm:containment}.
The degree-ordering is motivated by \cref{obs:degree_include},
which says that nodes of sufficiently large degree must be in $C$. 
Alternatively, one might order the nodes by the number of times they appear in a vertex cover.
Second, even though we are pruning the maximal matchings to be
minimal vertex covers, \cref{prop:core_overlap} provides motivation that the
matchings should be intersecting $C$. If only a constant number of nodes are
pruned when making the matching a minimal cover, then the overlap is still a
constant fraction of $C$.
Third, \cref{cor:matching_overlap} says that we shouldn't expect the union
to grow too fast.

We emphasize that the UMVC algorithm 
\emph{makes no assumption or use of the size of the planted cover $C$}.
Instead, we are only motivated by the theory of \cref{sec:theory}.
Finally, we note that there is a tradeoff
in computation and number of vertex covers. We chose 300 because it kept
the running time to about a minute on the largest dataset. However, a much
smaller number is needed to obtain the same recovery performance for some of our datasets.

%!TEX root = found-graph-data-paper.tex

\begin{figure}[p]
\newcommand{\skipsize}{-8pt}
  \centering
 \includegraphics[width=0.88\columnwidth]{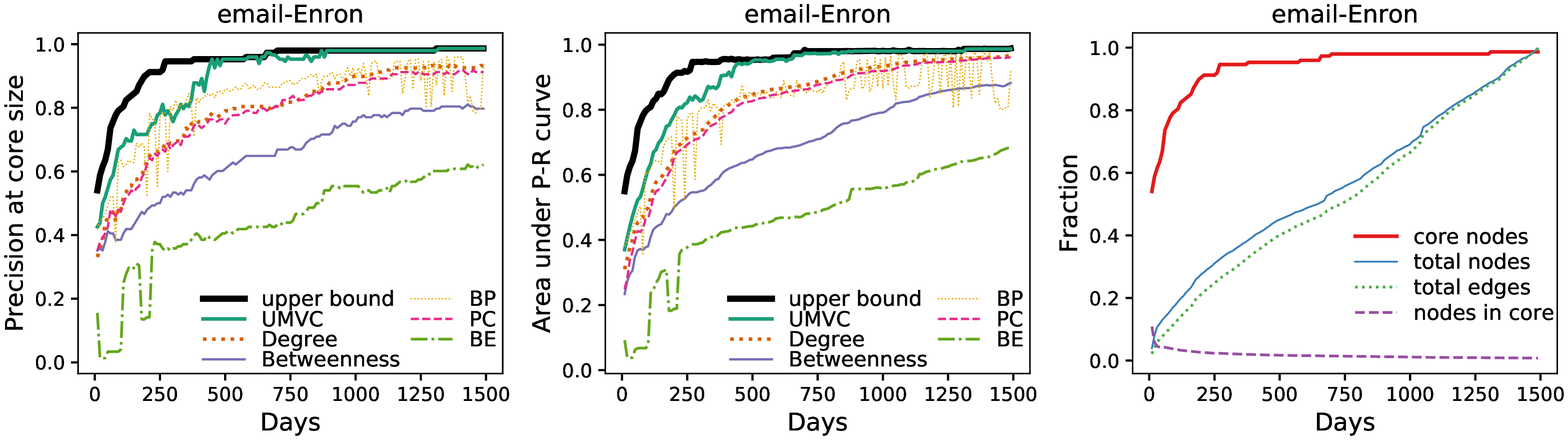} \vskip\skipsize
 \includegraphics[width=0.88\columnwidth]{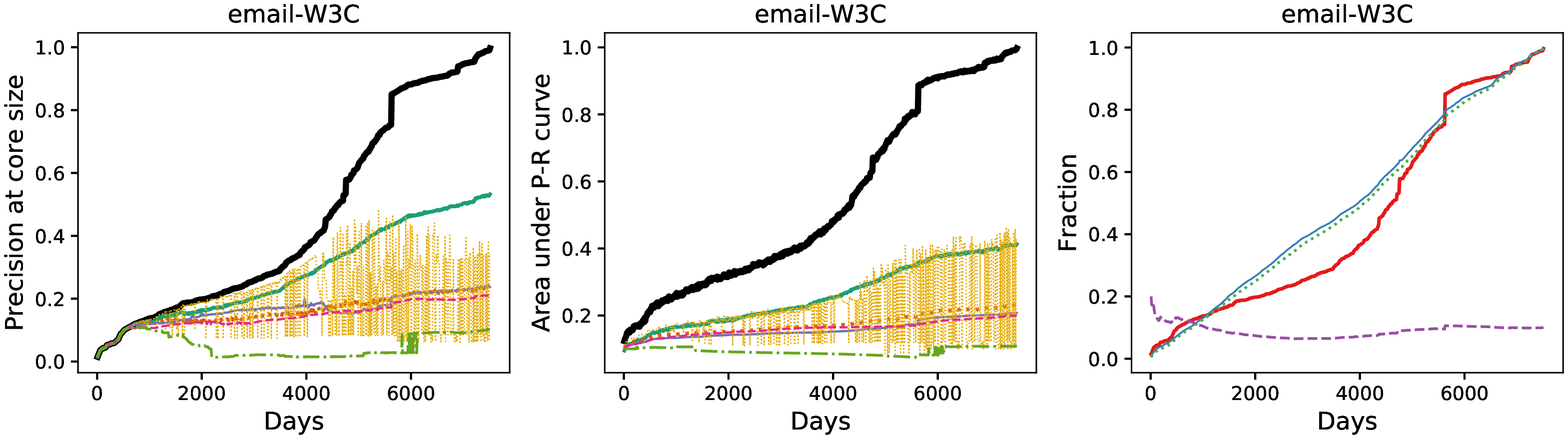} \vskip\skipsize
 \includegraphics[width=0.88\columnwidth]{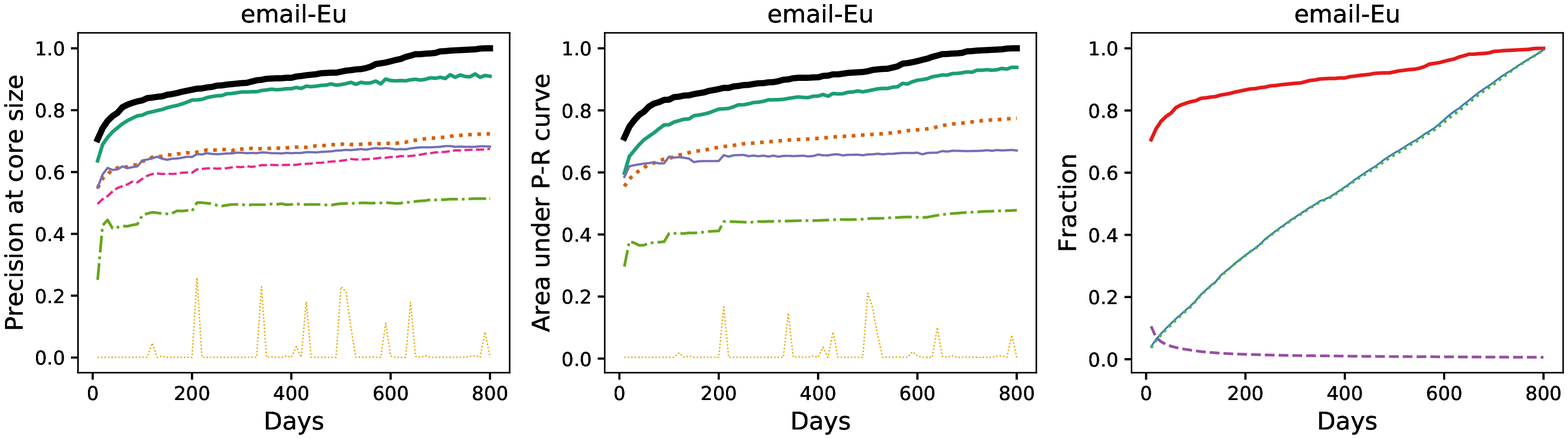} \vskip\skipsize
 \includegraphics[width=0.88\columnwidth]{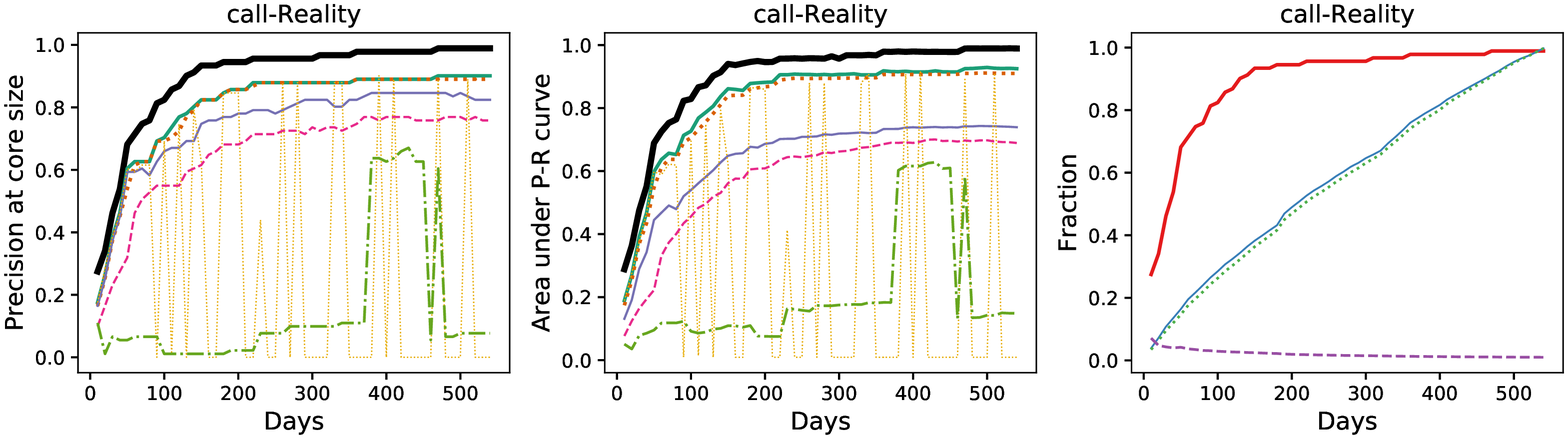}  \vskip\skipsize
 \includegraphics[width=0.88\columnwidth]{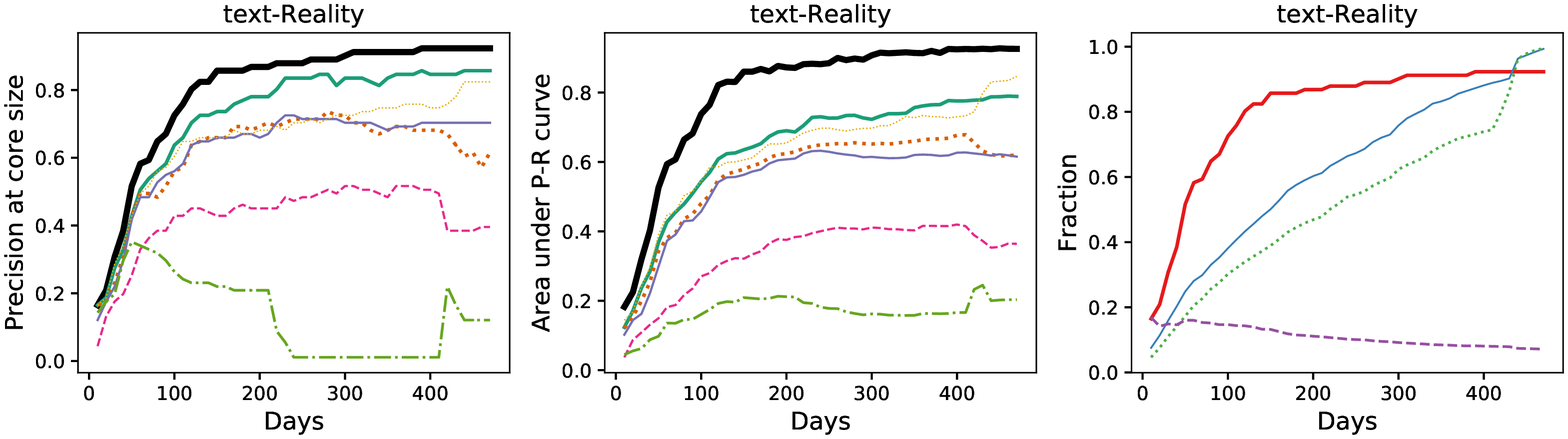}  
 \vspace{-1mm}
  \caption{Core recovery performance on five real-world datasets with 
  respect to six methods: 
  our proposed union of minimal vertex covers (UMVC), 
  degree ordering (degree), 
  betweenness centrality (betweenness;~\cite{Freeman-1977-betweenness}),
  belief propagation (BP;~\cite{Zhang-2015-SBM-CP}), 
  Borgatti-Everett scores (BE;~\cite{Borgatti-2000-CP}), 
  and Path-Core scores (PC;~\cite{Cucuringu-2016-Path-Core}). 
  Each dataset has timestamped edges, and we measure the performance of all 
  algorithms every 10 days of real time as the networks evolve. The left column shows 
  precision at the core size, the middle column is the area under the precision recall curve, 
  and the right column tracks the amount of data over time. Our proposed UMVC 
  algorithm performs well on all datasets. BP is sometimes competitive but is susceptible
  to poor local minima, as evidenced by its erratic performance on several datasets.}
  \label{fig:recovery}
\end{figure}

\xhdr{Other algorithms for comparison.}
We compare UMVC against five other methods. 
First, we consider an ordering of nodes by decreasing degree.
This heuristic captures the fact that the nodes outside of $C$ cannot link
to each other and that $\lvert C \rvert$ is much smaller than the total number
of vertices. This heuristic has previously been used as a baseline for
core-periphery identification~\cite{Rombach-2017-CP} and is theoretically justified
in certain stochastic block models of core-periphery structure~\cite{Zhang-2015-SBM-CP}.
Second, we use betweenness centrality~\cite{Freeman-1977-betweenness} to order the nodes, the idea being that nodes
in the core must appear in shortest paths between nodes in the fringe.
Third, we use the Path-Core (PC) scores~\cite{Cucuringu-2016-Path-Core} to order the nodes; these scores are a modified
version of betweenness centrality, which have been used to identify core-periphery structure in networks~\cite{Lee-2014-transport}.
Fourth, we use a scoring measure introduced by Borgatti and Everett (BE) for evaluating core-periphery structure with the core-fringe structure in which we are interested~\cite{Borgatti-2000-CP}.
Formally, the score vector $s$ is the minimizer of the function 
$\sum_{(i, j) \in E} (A_{ij} - s_is_j)^2$. We use a power-method-like iteration to compute $s$~\cite{Comrey-1962-minimum}.
Fifth, we use a belief propagation (BP) method designed for stochastic block models of core-periphery structure~\cite{Zhang-2015-SBM-CP}.

\xhdr{Results.}
We divide the temporal edges of each dataset into 10-day increments and construct
an undirected, unweighted, simple graph for the first $10r$ days of activity, $r = 1, 2, \ldots, \lfloor T / 10 \rfloor$, where $T$ is the total number of days spanned by the dataset (\cref{tab:summary_stats}). Given the ordering of nodes from each algorithm, we evaluate
recovery performance by the precision at core size (P@CS; \cref{fig:recovery}, left column) 
and the area under the precision recall curve (AUPRC; \cref{fig:recovery}, middle column).
We emphasize that no algorithm has knowledge of the actual core size.
We also provide an upper bound on performance, which for the P@CS metric
is the number of non-isolated nodes in the core at the time divided by the total
number of core nodes in the dataset and for AUPRC is the performance of an order
of nodes that places all of the non-isolated core nodes first and then a random order
for the remaining nodes. The right column of \cref{fig:recovery} provides statistics
on the growth of the network over time.

We observe that, across all datasets, our UMVC algorithm out-performs the 
degree, betweenness, PC, and BE baselines at essentially nearly all points in time.
The belief propagation sometimes exhibits better slightly performance but suffers
from erratic performance over time due to landing in local minima; for example,
see the recovery performance of the email-W3C dataset in row 2 of \cref{fig:recovery}.
In some cases, belief propagation can hardly pick up any signal, which is the case
in the email-Eu dataset (row 3 of \cref{fig:recovery}). In this dataset, UMVC clearly
out-performs other baselines. In the email-Enron dataset, UMVC achieves perfect
recovery after around 800 days of activity.

The key reason for the better performance of UMVC is that it uses the fact that
the core is a vertex cover. Low-degree nodes that might look like traditional ``periphery''
nodes in a core-periphery dataset but remain in the core are not picked up
by the other algorithms. Core-periphery detection algorithms in network science have 
traditionally relied on SBM benchmarks, eyeball tests, or heuristic benchmarks~\cite{Rombach-2017-CP,Zhang-2015-SBM-CP}. We have already theoretically
shown how the SBM induces substantial structure for this problem in our setup, and others
have performed similar analysis for block models~\cite{Zhang-2015-SBM-CP}; thus, this
may not be an appropriate benchmark for further analysis. Here, we give some notion
of ground truth labels on which to evaluate the algorithms and exploited the vertex cover
structure of the problem.

\subsection{Timing performance}

We also measure the time to run the algorithms on the entire dataset (\cref{tab:timing}).
For our UMVC algorithm, we use the union of 300 minimal vertex covers. 
Tuning the number of vertex covers provides a way for the application user to
trade off run-time performance and (potentially) recovery performance. The algorithm
was implemented in Julia (\cref{fig:julia_umvc}). 
The degree-based ordering was also implemented in Julia, betweenness centrality
was computed with the \texttt{LightGraphs.jl} julia package's implementation of Brandes'
algorithm for sparse graphs~\cite{Brandes-2001-betwenness,Brandes-2008-betwenness}.
The Path-Cores method was implemented in Python using the \texttt{NetworkX} library,
and the Belief Propagation algorithm was implemented in C++. 
We emphasize that our goal here is to demonstrate the approximate computation times, 
rather than to compare the most high-performance implementations possible.

The UMVC algorithm takes a few seconds for the email-W3C, email-Enron, reality-call, and reality-text datasets, and about one minute for the email-Eu dataset. This is an order
of magnitude faster than belief propagation (BP), and several orders of magnitude faster than betweenness centrality and Path-Core scores (PC). There are approximation algorithms for betweenness centrality, which would be faster than the exact 
algorithm~\cite{Bader-2007-betweenness,Geisberger-2008-betweenness}; however, the weak performance of the exact betweenness-based algorithm on our datasets did not justify our exploration of these approaches.

\begin{table}[tb]
\setlength{\tabcolsep}{5pt}
\centering
\caption{Time to run the algorithms on the largest dataset appearing in \cref{fig:recovery}. 
Our proposed union of minimal
vertex covers algorithm (UMVC) is fast 
and provides the best performance on several real-world datasets (see \cref{fig:recovery}).
Here, we use 300 minimal vertex covers in our UMVC algorithm; choosing the number of covers
allows the user to tune the running time.}
\begin{tabular}{l c c c c c c}
\toprule
Dataset & UMVC & degree 
& betweenness~\cite{Freeman-1977-betweenness}
& PC~\cite{Cucuringu-2016-Path-Core} 
& BE~\cite{Borgatti-2000-CP,Comrey-1962-minimum} 
& BP~\cite{Zhang-2015-SBM-CP} \\
\midrule
email-W3C    & 6.5 secs & $<$ 0.01 secs & 2.8 mins   & 1.1 hours  & \phantom{$<$} 0.1 secs     & 1.0 mins  \\ 
email-Enron  & 8.4 secs & $<$ 0.01 secs & 2.5 mins   & 1.8 hours  & \phantom{$<$} 0.1 secs     & 20.2 mins \\
email-Eu     & 1.2 mins & $<$ 0.01 secs & 11.8 hours & $>$ 3 days & \phantom{$<$} 0.9 secs     & 15.0 mins \\
call-Reality & 2.2 secs & $<$ 0.01 secs & 27.9 secs  & 6.1 mins   & \phantom{$<$} 1.8 secs     & 4.3 secs  \\
text-Reality & 0.5 secs & $<$ 0.01 secs & 0.8 secs   & 11.4 secs  & $<$ 0.1 secs & 6.8 secs  \\
\bottomrule
\end{tabular}
\label{tab:timing}
\end{table}

%!TEX root = found-graph-data-paper.tex

\section{Discussion}

Many network datasets are constructed with partial measurements and such data
is often ``found'' in some way that destroys the record of how the measurements
were made. Here, we have examined the particular case of
graph data where the edges are collected by observing all interactions involving some core
set of nodes, but the identity of the core nodes is lost. Such sets of core nodes
act as a planted vertex cover in the graph. In addition to developing
theory for this problem, we devised a simple and fast algorithm that recovers
such cores with extremely high efficacy in several real-world datasets.

There are a number of further directions that would be interesting to consider.
First, in order to develop theory and abstract the problem, 
we assumed that our graphs were
simple and undirected. 
However, there is much richer structure in 
the data that is generally collected. 
For example, the interactions in the email, phone call, and text messaging data are 
directional and could be modeled with a directed graph. Furthermore, we did not exploit the timestamps or the frequency of communication between nodes. One could incorporate this information into a weighted graph model of the data.

Substantial effort has been put forth in the network science community to study
mesoscale core-periphery structure in network data. However, the evaluation of such
methods has been empirical or evaluated on simple models such as the stochastic block
model, which we showed actually induces a substantial amount of structure on the recovery
problem. The present work
is the first effort to evaluate the recovery of 
core-periphery-like (i.e., core-fringe) network structure with much more worst-case assumptions
through the lens of machine learning with ``ground truth'' labels on the nodes. We hope that
this provides a valuable testbed for evaluating algorithms that reveal core-periphery structure,
although we should not take evaluation on ground truth labels as absolute~\cite{Peel-2017-groundtruth}.

Software accompanying this paper is available at: \\
\centerline{\url{https://github.com/arbenson/FGDnPVC}.}

%!TEX root = found-graph-data-paper.tex

\subsubsection*{Acknowledgments}
We thank 
Jure Leskovec for providing access to the email-Eu data;
Mason Porter and Sang Hoon Lee for providing the Path-Core code;
and Travis Martin and Thomas Zhang for providing the belief propagation code.
This research was supported in part by a Simons Investigator Award
and NSF TRIPODS Award \#1740822.

\bibliographystyle{abbrv}
\bibliography{refs}

\end{document}